\documentclass[preprint,12pt]{elsarticle}

\usepackage{graphics}

\usepackage{amssymb}
\usepackage{amsthm}
\usepackage{amsmath}

\journal{axiv.org}

\newcommand{\Zed}{\mathbb Z}

\newcommand{\Real}{\mathbb R}

\newtheorem{theorem}{Theorem}

\newtheorem{lemma}{Lemma}
\newtheorem{corollary}[theorem]{Corollary}
\newdefinition{remark}{Remark}
\newdefinition{example}{Example}
\newdefinition{definition}{Definition}

\begin{document}

\begin{frontmatter}



\title{A measure of state transition of collective of stateless automata in discrete environment}

\author{Igor Potapov}
\ead{potapov@liverpool.ac.uk}
\address{Department of Computer Science, University of Liverpool, Ashton Building, Ashton St, Liverpool L69 7ZF, U.K.}

\author{Oleksiy Kurganskyy\corref{cor1}}
\ead{kurgansk@gmx.de}
\address{Institute of Applied Mathematics and Mechanics, Ukrainian National Academy of Sciences, 74 R. Luxemburg St, Donetsk, Ukraine}
\cortext[cor1]{Corresponding author}

\begin{abstract}
This is an introductory paper in which we rise and
study two fundamental problems related to the 
analysis of a computational dynamic object distributed on 
the environment:
\begin{itemize} 
\item How to define unambiguously what is the state of such object? 
\item How to measure the amount of state transitions in this case?
\end{itemize} 
The main idea of the paper is to show that
the state of such computational dynamic object distributed on 
the environment can be described
by the language of internal and external states.
The results based on proposed approch 
have something in common with special relativity theory
and suggest existence of further connections between the automata theory and relativity theory.
\end{abstract}


\end{frontmatter}


\section{Introduction}

In this work a collective of interacting automata in an one-dimensional geometric 
environment is considered as an integral automata-like computational object.
The problem to define unambiguously what is the state of such dispersed and moving on the 
environment object and how to measure the amount of state transitions in this case
is quite non-trivial.
As opposed to the finite state automata where the measure of state transition 
is one state per unit of discrete time, for a computational dynamic object distributed 
on the environment different approaches to definition of the measure of state 
transition are possible. 

In this paper we propose a way for defining what a state is in the context of collective of
stateless automata. It allows us to define it on the basis of the relative positioning of 
automata, i.e. on the basis of its geometry.
The proposed approach distinguishes two types of states: internal and 
external states of automata collective. The measure of state transition of collective of automata, introduced 
in this paper, we name the proper time of this collective.

The proposed research is inspired by three major research directions which are: 
1) collective of automata in finite automata theory; 
2) discrete models of physical processes and projecting physical world into informational space of symbols and languages for computer modelling of physical world; 
3) studying the notion of time. 
Each of these research directions has an extensive bibliography confirming their importance~\cite{5,6,8,9,10}. 
The basis for this research is the notion of relativity as given by Poincare in his popular works~\cite{1,2,3,4}.

The concluding comparison of the obtained results with some formulas of special relativity theory 
shows that the formulated principles are invariant in relation to linguistic means of expression: semantic affinity of the principles (e.g., coordinate, velocity, reference frame) that form the language of our discrete model to the principles of language of special relativity theory resulted in their syntactic affinity (e.g., velocity-addition formula, ``length contraction/extension'' formula). The way of forming the language (velocity, time, reference frame) of interaction and interpreting this interaction between automata collectives in our model reflects Poincare's conventional point of view toward laws of physics. 

In order to suggest a physical analogy for this model we use the word ``body'' as alias for ``collective of automata''.

The paper is organized as follows. In the Section~\ref{body} we define the model of collective of automata. Then in the Section~\ref{state} we derive the notions of external and internal states of a collective of automata and study their properties connecting such notions as coordinate, spacial velocity, proper time velocity and proper time of collective.

\section{Body as a collective of automata}\label{body}

In what follows we use denotations $\Zed$ and $\Real$ 
for the sets of integers and real numbers, respectively. 
Also we denote the domains for the time and space coordinates by $T$ and $X$.
Initially in the model defintion we assume that $T$ and $X$ coincide with $\Zed$ 
but then we will extend it to $\Real$.

The general framework of the model, that we use for this study, 
consists of two main components: an environment $G$ that is represented by a graph and a set of stateless automata, which are interacting with the environment and between themselves.

The environment $G$ is defined as the infinite directed graph with the set of nodes $V=\{x+\frac{1}{2} | x \in \Zed \}$ and the set of edges $E= \{ (x-\frac{i}{2},x+\frac{i}{2}) | x \in \Zed, i \in \{-1,1\} \}$. An edge $(x-\frac{i}{2},x+\frac{i}{2})$ for some $i \in \{-1,1\}$ has the absolute coordinate $x\in \Zed$ and the direction $i$. Absolute coordinate of an edge $e$ will be denoted by $x(e)$ and its direction by  $r(e)$. Also the edge $e$ will be denoted by $x(e)^{r(e)}$. 
By the neighborhood of an edge $x^i$ we understand the pair of edges $x^i$ and $(x+i)^{-i}$. The edges $x^i$ and $(x+i)^{-i}$ will be called opposite edges and $x^i$ and $x^{-i}$ will be called contrary edges.

Stateless automata on the environment we name as elementary bodies. 
In the general framework we assume that elementary bodies are coloured in a way that isomorphic 
automata will have the same colour and non-isomorphic automata will have different colours. We assume that $r$ different numbered from $1$ to $r$ colours are used. Every moment of time $t$ any elementary body $b$ is located on an edge $b(t)$ of the graph $G$. 

The input for an elementary body, located on an edge $x^i$, is the sequence $(p_1,p_2,\ldots,p_r,q_1,q_2,\ldots,q_r)$ called the neighbourhood state of the edge $x^i$, where $p_k$ and $q_k$ are the numbers of elementary bodies of the colour $k$, 
located on the edges $x^i$ and $(x+i)^{-i}$ at the same moment of time, respectively.
The output of an elementary body is one of the two motions either the straight-line motion or the turn. 
If the output of an elementary body $b$ at a time moment $t\in\Zed$ on an edge $b(t)=x^i$ is the straight-line motion, then at the next time moment $b(t+1)=(x+i)^i$ and we say that it does not change its external state. If the output is the turn then $b(t+1)=x^{-i}$ and we say that the elementary body changes its external state.
Denoting by $\tau_b(t)$ the number of external state changes of $b$ until the moment of time $t$ we have that $1=\tau_b(t+1)-\tau_b(t)+\left|x(b(t+1))-x(b(t))\right|$ and also $t=\tau_b(t)-\tau_b(0)+s_b(t)$, where $s_b(t)= \sum^{t}_{j=1}|x(b(j))-x(b(j-1))|$ is the path covered by $b$ during the period of time $t$. 
In other words any elementary body uses the absolute time unit 
either for one spatial coordinate change in the environment or for one external state transition. We call $\tau = \tau_b(t)$ the proper time of $b$ and $w_b(t)=\tau_b(t+1)-\tau_b(t)$ the proper time velocity of $b$. 
Let us denote by $x_b(t)=x(b(t))$. We call $x_b(t)$ the absolute coordinate of $b$ at the moment of time $t$.
We denote by $v_b(t)=x_b(t+1)-x_b(t)$ the absolute spatial velocity of $b$ at the moment of time $t$. We call it uniform spatial velocity if $v_b(t)$ is a constant. For example, it follows from above definitions that any elementary body can have only one of the following uniform spatial velocities: $v=1$, $v=-1$, $v=0$.

The only state of a stateless automaton (i.e. of elementary body) we call the internal state. An elementary body is unambiguously defined by the set of input symbols that change its external state. In additional we assume also that elementary body can not change its external state anyway if its opposite edge is empty. 


We call the pair of a space coordinate $x$ and a time coordinate $t$ as coordinate in the 
absolute reference frame $O=X \times T $ and denote by the column vector. We call $O$ also the event space.
We define the discrete world line of $b$ in the event space from a time $t'$ to $t''$ as 
$b(t':t'')=\left\{\left( \begin{array}{c}x_b(t)\\t\end{array} \right)| t' \le t \le t'', t \in \Zed \right\}$, where $t',t''\in\Zed$, $t'\le t''$.

\begin{definition}
A body is an arbitrary finite set of elementary bodies.
\end{definition}

According to the defintion different bodies may have common parts and one body can contain another body as a subset.
If an elementary body belongs to a body then we will look at it as an elementary part of this body. An elementary body can be an elementary part of different bodies simultaneously.


The following two examples illustrate some of introduced definitions. 
Any elementary body in both examples changes its external state if and only if its opposite edge is not empty.
From it follows that all elemantary bodies are isomorphic.
We assume that all elementary bodies in each example are enumerated by integer numbers. 

\begin{example}\label{example1}
At time $t=0$ for each $x \in \Zed$ the elementary body with the number $x$ is located on the edge $x^{+1}$ if $x$ is even number and $x^{-1}$ otherwise. We define the body $A_1$ as the set $A_1=\{0,1,2\}$ of elementary bodies $0$, $1$ and $2$. 
\begin{figure}
\begin{center}
\includegraphics[scale=0.4]{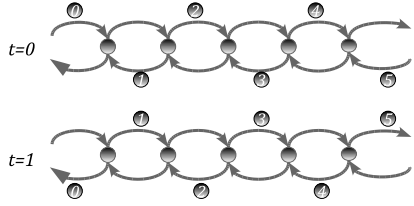}
\caption{Dynamics of elementary bodies from the Example~\ref{example1}}
\end{center}
\end{figure}
\end{example}

\begin{example}\label{example2}
At time $t=0$ for each $x \in \Zed$
the elementary body with the number $x$ has the coordinate $4 \left\lfloor \frac{x}{3} \right\rfloor+ (x \mod 3)$ and located on the edge with the direction $-1$ if $x \equiv 1 \mod 3$ and on the edge with the direction $+1$ otherwise. In this example we define the body $A_2=\{0,1,2\}$.
\end{example}

Let a body $B$ consist of $n$ elementary bodies enumerated by numbers 
$\{1,2, \ldots, n\}$. Then the absolute (average) coordinate of the body $B$ at time $t$ is the value
$x_B(t)=\frac{x_1(t)+ \ldots + x_n(t)}{n}$ and absolute spatial velocity of the body $B$ at time $t$ is the value 
$v_B(t)=x_B(t+1)-x_B(t)$. The bodies $A_1$ and $A_2$ from the above examples have 
uniform spatial velocities $0$ and $\frac{1}{3}$, respectively. From the definitions it follows that the
maximal possible positive or negative spatial velocities of any body can be $1$ or $-1$.

Since the coordinate values of a body can be non-integers let us
extend the absolute reference frame $O$ from $\Zed \times \Zed$ to $\Real \times \Real$.  
Let $t \in \Zed$ and $-\frac{1}{2}< \Delta \leq \frac{1}{2}$ then we say that an elementary body $b$  
at time $t+\Delta$ has the coordinate $x_b(t+\Delta)=x_b(t)+r(b(t))\cdot\Delta$ and is located on the edge 
$b(t+ \Delta)=b(t)$. 

%

Now we can define the (continues) world line $b(t':t'')$ of an elementary body $b$ in time interval from $t'$ to $t''$ as the extension of its discrete positions in the event space: $b(t':t'')=\left\{\left( \begin{array}{c}
x_b(t)\\t\end{array} \right)\left| t' \le t \le t'', t \in \Real\right. \right\}$, where $t',t''\in\Real$, $t'\le t''$. If $b(t':t'')$ is a straight line segment in the event space then the vector 
$\left( \begin{array}{c}
x_b(t'')-x_b(t') \\
t''-t'
\end{array} \right)$ can have only the direction either of the vector 
$\left( \begin{array}{c}
1\\
1
\end{array} \right)$ or of the vector $\left( \begin{array}{c}
-1\\
1
\end{array} \right)$ and we say that the $b(t':t'')$ corresponds to {\sl an elementary move} of $b$.


\section{Notion of state}\label{state}

In this section we define what does it mean that two bodies are in the same external or internal state, rather than what the external or internal state of a body in fact is. If needed the notion of state can be in generally defined as follows. Since the relation ``to be in the same external state'' is an equivalence relation, the external states can be defined as equivalence classes of this relation. The same holds for the definition of internal state.

\subsection{External state of a body}\label{external}

A body interacting with other bodies exert influence on them and at the same time is also under their
influence. It is quite natural to describe such influences on the basis of the notion of a state of a body. Our definition of a state of a body takes into consideration the relative positioning of its elementary parts in the environment. The changes of relative positioning of elementary parts in a body can affect the body entirely or a particular part of it. This motivates the question how to measure the amount of state transition. Before the definition of the notion of a state we introduce the denotation for the measure $\tau = \tau_B(t)$ of state transition of a body $B$ with the flow of time $t$. A casual meaning of $\tau = \tau_B(t)$ is the 
``age'' of the body $B$ at the moment $t$. We call $\tau = \tau_B(t)$ the proper time of $B$.

Independently from the definition of $\tau = \tau_B(t)$, we introduce the velocity $w_B(t)$ of state transition of the body $B$ as $w_B(t)=\tau_B(t+1)-\tau_B(t)$. We call this value as the proper time velocity of $B$ at the moment of the absolute time $t$.

%

\begin{definition}
For any body $B$ $w_B(t)=1 \Leftrightarrow \forall_{b \in B}w_b(t)=1$
\end{definition}

\begin{definition}
For any body $B$ $w_B(t)=0 \Leftrightarrow \forall_{b \in B}w_b(t)=0$
\end{definition}
From it follows that a body $B$ does not change its external state 
if all its elementary bodies do not change their external states. It means that two bodies are at the same external state if one can be transformed into another by 
isometric straight-line shifts in the environment applied to all its elementary parts.

\begin{theorem}
For any body $B$, if $|v_B(t)|=1$ then $w_B(t)=0$.
\end{theorem}
\begin{proof}
The statement follows from the fact that any change of the external state of a body 
is not possible in case of maximal spatial velocity of all its elementary parts.
\end{proof}

\subsection{Internal state of a body}

The notion of external state of a body allows to consider the bodies as an automata-like model of algorithms. 
But since two bodies with different absolute spatial velocities are definitely in different external states we can not speak of them as of realization of the same algorithm. For example there is no sense to ``ask'' a body to determine its absolute spatial velocity. 
However we would like to identify two bodies as the same algorithm even if they move with different spatial velocities.
It will be achieved by introduction of affine isomorphism of bodies through definition of inertial reference frame associated with a body so that the external state of a body will be presented as pair of components: spatial velocity of the body and its internal state that is spatial velocity invariant. 
The point of introducing the notion of inertial  reference frame 
associated with a body lies in the ability to consider other bodies in 
relation to the given one.
An example of inertial reference frame is the absolute reference frame $O$ associated with an immovable body $B$ such that for all $t\in\Zed$ $x_B(t)=0$, $v_B(t)=0$, $w_B(t)=1$, $\tau_B(0)=0$, and, hence, $\tau_B(t)=t$. Thus, the introduced notions of absolute time, absolute coordinate and absolute spatial velocity implicitly mean an absolutely motionless body in relation to which objects were considered. The reference frames associated with the bodies allow us to make these notions relative.

Let us denote (for a pair of bodies $A$ and $B$) by $x_{AB}(\tau_B)$, $v_{AB}(\tau_B)$, $w_{AB}(\tau_B)$ and $\tau_{AB}(\tau_B)$ the coordinate, the spatial velocity, the proper time velocity and the proper time of the body $A$ at the moment of time $\tau_{B}$ in the reference frame $O_B$ associated with the body $B$, respectively.
By definition we assume that $x_{BB}(\tau_B)\equiv 0$, $v_{BB}(\tau_B)\equiv 0$, $w_{BB}(\tau_B)\equiv 1$ and $\tau_{BB}(\tau_B)=\tau_B$. 

\begin{definition}
A body $B$ is called an inertial body if $v_B(t)$ and $w_B(t)$ are both constants.
\end{definition}

\begin{remark}\label{rem1}
It follows that $\tau_{AB}(\tau_B)=\tau_{AB}(0)+\tau_B\cdot w_{AB}$ and $x_{AB}(\tau_B)=x_{AB}(0)+\tau_B\cdot v_{AB}$ for inertial bodies $A$ and $B$.
\end{remark}

\begin{remark}
Further we consider only inertial bodies.
\end{remark}

In addition we assume that coordinates of the same events in different inertial reference frames are connected by affine mapping. For any bodies $A$ and $B$ let us denote by $L_{BA}:O_B\rightarrow O_A$ the affine mapping that connects $O_B$ and $O_A$ such that an event $(x,\tau_B)$ in $O_B$ coincides with the event $L_{BA}(x,\tau_B)$ in $O_A$.

These assumptions are sufficient to find out $L_{BA}$. Without loss of generality we assume that the origins of both reference frames $O_A$ and $O_B$ are the same: $x_{BA}(0)=0$ and $\tau_{BA}(0)=0$. Then the mapping $L_{BA}$ is linear. Let us work out the form of transformation matrix 
$L_{BA}=\left( \begin{array}{cc}
a_{11} & a_{12}  \\
a_{21} & a_{22} \end{array} \right)$. 

\begin{lemma}\label{lemma_LBA}
The mapping $L_{BA}$ either holds the directions of the vectors $\left(\begin{array}{c} 1\\1\end{array}\right)$
and $\left(\begin{array}{c} -1\\1\end{array}\right)$ 
(i.e. these vectors are eigenvectors of the mapping $L_{BA}$) 
or permutes their directions.
\end{lemma}
\begin{proof}
The directions of reference frame axes are imaginary directions in the event space. But the set of directions of the vectors $\left(\begin{array}{c} 1\\1\end{array}\right)$ and $\left(\begin{array}{c} -1\\1\end{array}\right)$ in the absolute reference frame corresponds to the directions of the ``real'' ``material'' world lines of elementary bodies by elementary moves and therefore this set of directions does not depend on reference frames. From it follows that this set of directions {\sl is invariant} by any affine transformation.
\end{proof}

\begin{corollary}
For the matrix $L_{BA}$ holds either $a_{11}=a_{22}$, $a_{12}=a_{21}$ or $a_{11}=-a_{22}$, $a_{12}=-a_{21}$.
\end{corollary}
\begin{proof}
Based on the lemma~\ref{lemma_LBA} the corollary statement follows as a result of straightforward calculations.
\end{proof}

\begin{definition}
If $a_{11}=a_{22}$ and $a_{12}=a_{21}$ holds, then the reference frames $O_A$ and $O_B$ are said to be in standard configuration. If $a_{11}=-a_{22}$ and $a_{12}=-a_{21}$ holds, then the reference frames $O_A$ and $O_B$ are said to be in symmetric configuration.
\end{definition}

\begin{theorem}\label{theorem1} In the standard configuration the following holds
$$L_{BA}=
\left( \begin{array}{cc}
1/w_{BA} & v_{BA}/w_{BA}  \\
v_{BA}/w_{BA} & 1/w_{BA} \end{array} \right).
$$
\end{theorem}
\begin{proof}

Since 
$L_{BA} \cdot 
\left( \begin{array}{c}
x_{BB}(\tau_{BA}(\tau_A))  \\
\tau_{BA}(\tau_A)
\end{array} \right)
=
\left(\begin{array}{c}
x_{BA}(\tau_A)\\
\tau_A\end{array}\right)
$, 
$x_{BA}(\tau_A)=v_{BA}(\tau_A)\cdot\tau_A$, 
$\tau_{BA}(\tau_A)=w_{BA}(\tau_A)\cdot\tau_A$,
$x_{BB}(\tau_B) \equiv 0$, 
 then 
$L_{BA} \cdot 
\left( \begin{array}{c}
0  \\
1
\end{array} \right)
=
\left( \begin{array}{c}
v_{BA}/w_{BA}  \\
1/w_{BA} \end{array} \right)
$.
From it follows that $a_{12}=v_{BA}/w_{BA}$ and $a_{22}=1/w_{BA}$.
\end{proof}

\begin{corollary}
In the symmetric configuration the following holds
$$L_{BA}=
\left( \begin{array}{cc}
-{1/w_{BA}} & v_{BA}/w_{BA}  \\
-{v_{BA}/w_{BA}} & 1/w_{BA} \end{array} \right) .
$$
\end{corollary}
\begin{proof}
A proof is analogous to that of Theorem~\ref{theorem1}.
\end{proof}

\begin{corollary}
Let $v_{AB}=0$. By symmetric configuration the space axis of the reference frames $O_A$ and $O_B$ are in opposite directions, by standard configurations they are in the same directions.
\end{corollary}

Further for the sake of convenience we consider reference frames only in the standard configuration.
The following corollaries hold for any intertial bodies $A$, $B$, $C$.

\begin{corollary}\label{cor3}
It holds $v_{AB}=-v_{BA}$ and $w_{AB}\cdot w_{BA}=1-v_{AB}^2=1-v_{BA}^2.$
\end{corollary}
\begin{proof}
The equalities can be derived from $L_{AB}\cdot L_{BA}= 
\left( \begin{array}{cc}
1 & 0  \\
0 & 1 \end{array} \right)
$.
\end{proof}

\begin{corollary}\label{cor_addition}
(velocity-addition formula)
$v_{CA}=\frac{v_{BA}+v_{CB}}{1+v_{BA}v_{CB}}$.
\end{corollary}
\begin{proof}
This velocity-addition formula is derived from the equation $L_{CA}=L_{BA}\cdot L_{CB}$.
\end{proof}

\begin{corollary}\label{corollary1}
(``length contraction/extension'' formula)
Given inertial bodies $A$, $B$ and $C$ such that $v_{AC}=v_{BC}$. Let $\Delta x=\left|x_{AA}(\tau_A)-x_{BA}(\tau_A)\right|$ be the distance between $A$ and $B$ in the reference frame $O_A$. Let $\Delta x'=\left|x_{AC}(\tau_C)-x_{BC}(\tau_C)\right|$ be the distance between $A$ and $B$ in the reference frame $O_C$, then $\Delta x'=w_{CA}\cdot \Delta x$.
\end{corollary}
\begin{figure}
\begin{center}
\includegraphics[scale=0.6]{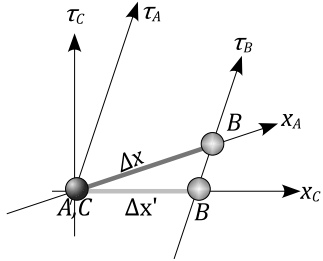}
\caption{Distances between two bodies $A$ and $B$ in Corollary~\ref{corollary1}}
\end{center}
\end{figure}
\begin{proof}
Notice that the values of $\Delta x$ and $\Delta x'$ are constants. Without loss of generality we assume $\tau_{AC}(0)=\tau_{BC}(0)=0$, $x_{AC}(0)=0$, $v_{AC} \geq 0$, $x_{BA} \geq 0$. Then $x_{BA}(\tau_A)\equiv \Delta x$ and $x_{BC}(0)=\Delta x'$. Let $\tau_A$ be such a moment of time that the events $(x_{BC}(0),0)=(\Delta x',0)$ and $(x_{BA}(\tau_A),\tau_A)=(\Delta x,\tau_A)$ are the same. Then the formula $\Delta x'=w_{CA}\cdot \Delta x$ of 
``length contraction'' follows from $L_{CA}\cdot\left(\begin{array}{c} \Delta x' \\ 0\end{array} \right)=\left(\begin{array}{c} \Delta x \\ \tau_A\end{array} \right)$ and Theorem~\ref{theorem1}.
\end{proof}

As it will be seen, from the example at the end of this section, $w_{CA}$  may take on a value which is 
less than 1 as well as more than 1. So it means that in our discrete model 
we have contracting length as well as extending length in respect to 
different inertial frame system.

Now we give a definition of internal state of a body. Let for bodies $A$ and $B$ there be a bijection $\phi:A\rightarrow B$ such that for all $b\in A$ elementary bodies $b$ and $\phi(b)$ are isomorphic. 
We say that $A$ at the moment of proper time $\tau_A$ and $B$ at the moment of proper time $\tau_B$ are affine isomorphic iff $\{(\phi(b),x_{bA}(\tau_A)|b\in A\}$=$\{(b,x_{bB}(\tau_B)|b\in B\}$. 
\begin{definition}
Two inertial bodies are in the same internal state at some moments of their proper time iff they are affine isomorphic at their respective proper time.
\end{definition}

Internal state of an inertial body does not depend on its spatial velocity in the absolute reference frame. Thus, the external state of an inertial body can be seen as a combination of two components: the spatial velocity of the body and its internal state.

In order to illustrate the concept of affine isomorphism let us consider bodies $A_1$ and $A_2$ from the Examples~\ref{example1} and~\ref{example2}. This bodies are affine isomorphic. The corresponding transformation of the reference frame $O_2$ of $A_2$ to $O_1$ of $A_1$ is:

\[ \left( \begin{array}{c}
x'  \\
t' \end{array} \right)
=
\left( \begin{array}{cc}
\frac{3}{2} & \frac{1}{2}  \\
\frac{1}{2} & \frac{3}{2} \end{array} \right)
\left( \begin{array}{c}
x \\
t \end{array} \right)
-
\left( \begin{array}{c}
\frac{1}{2}  \\
\frac{1}{2} \end{array} \right).\]

The dynamics of the bodies and illustration of the transformation are shown on the Figure~\ref{F_ex2}. From the value of transformation matrix and Corollary~\ref{cor3} it follows that $v_{21}=-v_{12}=1/3$ $w_{21}=2/3$, $w_{12}=4/3$.
\begin{figure}
\begin{center}
\includegraphics[scale=0.4]{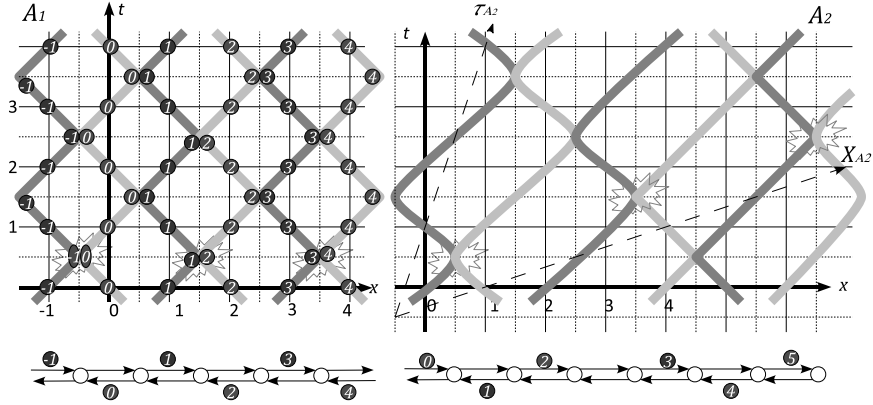}
\caption{The time-space diagrams for the collectives of automata from
Examples~\ref{example1} and~\ref{example2}. }\label{F_ex2}
\end{center}
\end{figure}

\section{Final Remarks and Conclusion}\label{conclusion}

Let us compare the obtained results with formulas of special relativity theory.
It is interesting to have a look, from our model viewpoint, at two equations $\Delta t'={\Delta t}/\sqrt{1-(v/c)^2}$ of time dilation and $\Delta x'={\Delta x}\cdot\sqrt{1-(v/c)^2}$ of length contraction of the special relativity theory. Drawing a proper analogy between them and $\tau_{AC}(\tau_C)-\tau_{AC}(0)=w_{AC}\cdot\tau_C$ (Remark~\ref{rem1}) and $\Delta x'=w_{CA}\cdot \Delta x$ (Corollary~\ref{corollary1}) respectively we can see, due to generally asymmetry $w_{AC}\neq w_{CA}$ in our discrete virtual ``world'', that the coefficient $1/\gamma=\sqrt{1-(v/c)^2}$ reciprocal to Lorentz factor $\gamma$ has different ``physical'' meanings in these formulas. The factor $1/\gamma$ has in the first equations a meaning of the coefficient $w_{AC}$ and in the second equations has a meaning of the coefficient $w_{CA}$ if we consider a ``moving'' $A$ with respect to a ``rest'' $C$.

We would like to position this paper as an introductory research work on
fundamental notion of a state for distributed automata object and to draw attention
to the number of problems related to such notion. It is shown that
a measure of state transition of such object  can be described
by the language of internal and external state changes. 
We hope that the proposed analogies between automata theory and relativity theory 
can generate further interest to the topic towards a better understanding of
such analogy.

Apart from the study of the notion of state 
somebody can ask a number of more technical questions 
which were not intention of this work, but nevertheless 
are important research issues.
In particular, it was not considered what kind of values the transformation matrix of $L_{AB}$ basically can have. 
Also the algorithmic universality of model is not proved, though a proof of this fact simply enough by simulation of cellular automata.  It will be interesting to consider the model in higher dimensions and the case of not inertial bodies and inhomogeneous environment.
At the same time various problems are unsolvable in the given model because of peculiarity of the model, e.g. like the question whether a body can define its absolute velocity. This seemingly natural question is meaningless in the considered model and therefore is an algorithmically unsolvable problem in it.  These and a number of other questions will be considered in the future publications of authors. \\

{\bf Acknowledgements}: 
The author acknowledges the useful discussions on this work with Dr. Valeriy Kozlovskyy
and would like to thank for his valuable comments that  helped us to improve the presentation. 
The work of authors was supported in part by  NATO Collaborative Linkage Grant 983162.





\bibliographystyle{elsarticle-num}



\end{document}